\documentclass[letterpaper, 10 pt, conference]{ieeeconf}
\IEEEoverridecommandlockouts
\usepackage{color}
\usepackage{graphicx}          
\usepackage{amsmath}
\usepackage{amssymb}                               
\usepackage{mathtools}
\usepackage{blindtext}
\usepackage{booktabs, caption, subcaption, subfloat,algorithm,algorithmic, psfrag}
\usepackage{url,units}
\renewcommand{\QED}{\QEDopen}

\newtheorem{assumption}{Assumption}
\newtheorem{definition}{Definition}
\newtheorem{theorem}{Theorem}
\newtheorem{remark}{Remark}
\newtheorem{lemma}{Lemma}
\newtheorem{corollary}{Corollary}
\newtheorem{proposition}{Proposition}

\newcommand{\eq}{\ensuremath{_{\text{eq}}}}

\newcommand{\lqr}{\ensuremath{_{\text{lqr}}}}
\newcommand{\inv}{\ensuremath{_{\text{inv}}}}
\newcommand{\xin}{\ensuremath{\mathcal{X}_{\text{in}}}}

\newcommand{\rnx}{\ensuremath{\mathbb{R}^{n_x}}}
\newcommand{\rnu}{\ensuremath{\mathbb{R}^{n_u}}}
\newcommand{\hm}{\ensuremath{\hphantom{-}}}

\usepackage{xcolor}
\usepackage{siunitx}
\usepackage{interval}

\title{\LARGE \bf Stability and feasibility of neural network-based controllers via output range analysis}

\author{ Benjamin Karg and Sergio Lucia
\thanks{B. Karg and S. Lucia are with the Laboratory of Internet of Things for Smart Buildings,
Technische Universit\"{a}t Berlin, Einsteinufer 17, 10587 Berlin, Germany and Einstein Center Digital Future.
(e-mail: {\tt\small \{benjamin.karg, sergio.lucia\}@tu-berlin.de)}}
}

\begin{document}
\maketitle
\thispagestyle{empty}
\pagestyle{empty}
\begin{abstract}

Neural networks can be used as approximations of several complex control schemes such as model predictive control. We show in this paper which properties deep neural networks with rectifier linear units as activation functions need to satisfy to guarantee constraint satisfaction and asymptotic stability of the closed-loop system. To do so, we introduce a parametric description of the neural network controller and use a mixed-integer linear programming formulation to perform output range analysis of neural networks.
We also propose a novel method to modify a neural network controller such that it performs optimally in the LQR sense in a region surrounding the equilibrium. 
The proposed method enables the analysis and design of neural network controllers with formal safety guarantees as we illustrate with simulation results.
\end{abstract}


\section{Introduction}

The success stories of deep learning have recently motivated the idea of using neural networks as controllers, although it was already proposed to approximate nonlinear model predictive control in~\cite{parisini1995receding}.
Two main methods exist for deriving neural network-based controllers: reinforcement learning (RL)~\cite{recht2019tour} and imitation learning (IL).
Reinforcement learning directly computes optimal control policies by \emph{trying} different inputs on a simulator of the reality and computing the obtained reward, which is then used to compute the next inputs.  RL was able to reach super-human performance in playing arcade games~\cite{mnih2015human}.
While the methods for arcade games only handled discrete control inputs, different methods have been proposed~\cite{lillicrap2015continuous} to extend the application of RL to continuous control inputs. Recently, other works have considered probabilistic safe learning and exploration of RL-based controllers using ideas from safety tubes, chance constraints and reachability of terminal sets~\cite{wabersich2019probabilistic,koller2018learning}.

Neural networks derived via imitation learning are usually obtained by minimizing the error between the output of the network and an optimal control input, which can be obtained by solving an optimization problem. This strategy has been shown to achieve comparable performance to exact optimal control schemes, but usually only probabilistic safety certificates can be derived including information of the dual problem~\cite{zhang2019safe} or using probabilistic verification approaches approaches~\cite{karg2019probabilistic,hertneck2018learning,nubert2020safe,wabersich2018linear,karg2019learning}.
A different alternative to achieve guarantees of neural network-baesd approximate controllers is to perform an additional projection onto the feasible set~\cite{chen2018approximating,karg2018efficient} or use the network to compute the initial guess for an optimization solver~\cite{chen2019large} instead of using it as a controller.

To obtain deterministic guarantees about the performance of neural network controllers, methods from output range analysis were recently used for the analysis of the behaviour of the controlled system~\cite{dutta2018learning,dutta2019reachability,xiang2018reachable}.
By predicting the reachable set of the controlled system, safe application of the neural network controller can be stated or rejected.
Formulations providing exact bounds for the output range of neural networks are satisfiability (SAT)~\cite{pulina2012challenging}, satisfiability modulo theory (SMT)~\cite{scheibler2015towards,katz2017reluplex,huang2017safety,ehlers2017formal} and mixed-integer programming (MIP)~\cite{dutta2018output,anderson2020strong}.
Because these formulations result in $\mathcal{NP}$-hard problems, methods with a reduced computational load giving approximate bounds based on semi-definite programming (SDP)~\cite{fazlyab2019safety,fazlyab2019probabilistic} and other relaxations~\cite{dvijotham2018dual,salman2019convex,lomuscio2017approach} were developed.

In this work, by combining methods from output range analysis with notions from optimal control design, we show how closed-loop constraint satisfaction and asymptotic stability for neural network controllers can be guaranteed, without the need for additional strategies such as a projection to the feasible space or a backup controller.
The main novel contributions of this work are (i) the derivation of requirements for neural network controllers to guarantee asymptotic stability,  and (ii) the extension of an MILP formulation to verify closed-loop constraint satisfaction and asymptotic stability, and (iii) an optimization-based modification of a previously trained neural network-controller to ensure that  the closed-loop satisfies the requirements for asymptotic stability and constraint satisfaction.

The remainder of the paper is organized as follows.
Section~\ref{sec:problem_statement} states the considered problem and 
the MILP formulation for the verification of the neural network controllers is presented in Section~\ref{sec:MILP}.
In Section~\ref{sec:safety}, the necessary requirements for the safe application of the neural network controller are derived.
Section~\ref{sec:NN_LQR} shows how a neural network controller can be modified, such that it performs like a LQR controller in an environment containing the equilibrium point.
The contributions of this work will be demonstrated and visualized with a case study in Section~\ref{sec:case_study} and the paper is concluded in Section~\ref{sec:conclusions}.

\section{Problem statement}\label{sec:problem_statement}

\subsection{Notation}
The set of natural numbers $\{1, \dots, L\}$ with $L \in \mathbb{N}$ is denoted by $[L]$.
The real scalars, vectors and matrices are denoted by $\mathbb{R}$, $\mathbb{R}^n$ and $\mathbb{R}^{n \times m}$, respectively, with $n, m \in \mathbb{N}$.
The positive semi-definite and positive definite real matrices of dimension $n \times n$ are given by $\mathbb{S}_{+}^{n}$ and $\mathbb{S}_{++}^{n}$, respectively.
A superscript in parentheses with two elements refers to the element in  the $i$-th row and the $j$-th column of a matrix, e.g. $W^{(i,j)}$, and the $i$-th row of a matrix $W^{(i)}$ or a vector $b^{(i)}$ is denoted by a single element in parentheses in the superscript.
The interior of a set is denoted by $\text{int}(\cdot)$ and the composition of two functions $f$ and $g$ is given by $f \circ g(\cdot) = f(g(\cdot))$.
The symbol $\odot$ denotes element-wise multiplication of two scalars, vectors or matrices of same dimension.
The vector of ones is denoted by $\mathbf{1}$ and $I$ is the identity matrix of corresponding size. 
The term $\mathcal{Y} = f(\mathcal{X})$ is short for $\mathcal{Y} \coloneqq \{y \in \mathbb{R}^{n_y} \, | \, y = f(x), \, \forall \, x \in \mathcal{X}\}$.

\subsection{Optimal control}
We consider linear time-invariant systems
\begin{align}\label{eq:lti}
    x_{k+1} = A x_k + B u_k,
\end{align}
where $x_k \in \rnx$ are the states and $u_k \in \rnu$ are the inputs at time step $k$, $A \in \mathbb{R}^{n_x \times n_x}$ is the state matrix, $B \in \mathbb{R}^{n_x \times n_u}$ is the input matrix and the pair $(A,B)$ is stabilizable.
The classical goal of optimal control is to drive the system to the equilibrium $x\eq$, which is assumed to be at the origin w.l.o.g., by minimizing the objective
\begin{align}\label{eq:objective_oc}
    J = \sum_{k=0}^{\infty} x_k^T Q x_k + u_k^T R u_k,
\end{align}
where $Q \in \mathbb{S}_{+}^{n_x}$ and $R \in \mathbb{S}_{++}^{n_u}$ are the state and input weight matrices.
If the optimal control problem (OCP) is unconstrained, the optimal solution is given by the discrete-time infinite horizon linear-quadratic regulator (LQR), where a state feedback law 
\begin{align}\label{eq:lqr_feedback}
 u_k = -K\lqr x_k,
\end{align}
is applied.
If polytopic state constraints $\mathcal{X} = \{ x \in \rnx \, | \, C_x x \leq c_x  \}$, with $C_x \in \mathbb{R}^{n_{cx} \times n_x}$, $c_x \in \mathbb{R}^{n_{cx}}$, and input constraints $\mathcal{U} = \{ u \in \rnu \, | \, C_u u \leq c_u  \}$, with $C_u \in \mathbb{R}^{n_{cu} \times n_u}$, $c_u \in \mathbb{R}^{n_{cu}}$, are present, with $x\eq \in \text{int}(\mathcal{X})$ and $0 \in \text{int}(\mathcal{U})$, the violation-free application of the LQR feedback law is reduced to a region around the equilibrium state $x\eq$  where no constraints are active.
The LQR admissible region $\mathcal{R}\lqr$ is described as a polytopic set~\cite{bemporad2002explicit}:
\begin{align}\label{eq:lqr_set}
    \mathcal{R}\lqr \coloneqq \{x \in \rnx \, | \, F\lqr x \leq g\lqr\},
\end{align}
where $F\lqr \in \mathbb{R}^{m\lqr \times n_x}$ and $g\lqr \in \mathbb{R}^{m\lqr}$.

In order to derive a controller satisfying the constraints and performing optimally outside of this region, a model predictive control (MPC)~\cite{rawlings2009model} scheme can be formulated:
\begin{subequations}\label{eq:mpc}
\begin{align}
    & \underset{\mathbf{x},\mathbf{u}}{\text{minimize}} & & x_N^T P x_N + \sum_{k = 0}^{N-1} x_k^T Q x_k + u_k^T R u_k  \\
& \text{subject to} & & x_0 = x_{\text{init}}, \\
&&& x_{k+1} = A x_k + B u_k, \quad \forall k \in [N],  \\
&&& x_k \in \mathcal{X}, \, \qquad \qquad \qquad \forall k \in [N], \\
&&& u_k \in \mathcal{U}, \, \, \quad \quad \qquad \qquad \forall k \in [N-1], \\
&&& x_N \in \mathcal{R}\lqr, 
\end{align}
\end{subequations}
where $N$ is the horizon of the MPC scheme, $P \in \mathbb{R}^{n_x \times n_x}$ is the terminal weight matrix and $\mathbf{x} = [x_0^T, \dots, x_{N}^T]^T$ and $\mathbf{u} = [u_0^T, \dots, u_{N-1}^T]^T$ are the concatenation of the state and control decision variables.
Problem~\eqref{eq:mpc} is solved in every control instant and the optimal control input $u_0^*$ is applied to the system~\eqref{eq:lti}.
The goal of many learning-based controllers is to approximate the solution of~\eqref{eq:mpc}, which is a piece-wise linear function mapping the initial states $x_{\text{init}}$ to an input $u_0^*$~\cite{bemporad2002explicit}.

\subsection{ReLU networks}
By using a feed-forward neural networks $\mathcal{N}: \rnx \rightarrow \rnu$ with ReLU activation functions, the solution of~\eqref{eq:mpc} is approximated~\cite{karg2018efficient}, where each input of the network is a state $x$ and the output of the network corresponds to the control input $u_0$.
A neural network has $L$ hidden layers with $n_l$, $l \in [L]$, neurons per hidden layer, $n_x$ input neurons, $n_u$ output neurons and is defined as the function composition:
\begin{align}\label{eq:neural_network}
\begin{split}
\mathcal{N}(x;\theta) = \qquad \qquad \qquad \qquad \qquad \qquad \qquad \qquad \qquad\\
\quad \bigg \{
\begin{array}{lll}
		f_{L+1} \circ \sigma_L \circ f_L \circ \dots \circ \sigma_1 \circ f_1(x) & \text{for} & L \geq 2, \\
		f_{L+1} \circ \sigma_1 \circ f_1(x), & \text{for} & L = 1,
\end{array}
\end{split}
\end{align}
where each $f_l$ is an affine transformation of the output of the previous layer:
\begin{align}
    f_l(\xi_{l-1}) = W_l \xi_{l-1} + b_l,
\end{align}
where $\xi_{l-1} = W_{l-1} \xi_{l-2} + b_{l-1}$ for $l = 2, \dots, L+1$ and the initial input to the first layer is a state $\xi_0 = x$.
The nonlinear ReLU function $g_l$ returns the element-wise maximum between zero and the affine function of a layer $l$:
\begin{align}\label{eq:relu}
    \sigma_l(f_l) = \text{max}(0,f_l).
\end{align}
The set of parameters $\theta = \{\theta_1,\dots,\theta_{L+1}\}$, with $\theta_l = \{W_l, b_l\}$, contains the weights
\begin{align}
W_l \in
	\begin{cases}
		\mathbb{R}^{n_1 \times n_x} & \text{if} \quad l=1, \\
		\mathbb{R}^{n_l \times n_{l-1}} & \text{if} \quad l=2,\dots,L, \\
		\mathbb{R}^{n_u \times n_L} & \text{if} \quad l=L+1,
	\end{cases}
\end{align}
and biases
\begin{align}
b_l \in
	\begin{cases}
		\mathbb{R}^{n_l} & \text{if} \quad l=1,\dots,L, \\
		\mathbb{R}^{n_u} & \text{if} \quad l = L+1,
	\end{cases}
\end{align}
which define the piece-wise affine function that the neural network describes.

By viewing the neurons as hyperplanes~\cite{hanin2019deep}, an activation pattern can be defined, which assigns a binary value to every neuron in  the hidden layer to model the ReLU function~\eqref{eq:relu}.
This limits the maximum possible number of different activation patterns to
\begin{align}\label{eq:max_number_activation_patterns}
    n_{\text{act,max}} = 2^{n_{\text{neurons}}},
\end{align}
where $n_{\text{neurons}} = \sum_{i=1}^{L}n_l$.
Each activation pattern is a collection of $\Gamma_i = \{\gamma_{i,1}, \dots, \gamma_{i,L} \}$, $i \leq n_{\text{act,max}}$, of $L$ vectors $\gamma_{i,l} \in [0, 1]^{n_l}$, one for each hidden layer.  

The activation pattern corresponding to a state $x$ can be derived via:
\begin{align}\label{eq:activation_vector}
     G(x) &\coloneqq \left\{ \beta \circ f_l(\xi_{l-1}) \in [0,1]^{n_l} \, | \, l \in [L], \xi_0 = x \right\}.
\end{align}
The function $\beta(\cdot)$ converts each element $i$ of the output of a layer to a binary representation based on the output of the previous layer:
\begin{align}\label{eq:beta}
    \beta \circ f_l(\xi_{l-1})^{(i)} =\begin{cases}
                    1 &\text{if } W_l^{(i)} \xi_{l-1} + b_l^{(i)} \geq 0,\\
                    0 &\text{else}.
              \end{cases}
\end{align}
Each activation pattern $\Gamma_i$ implicitly describes a polytopic region in the state space via
\begin{align}
    \mathcal{R}_{\Gamma_i} = \{ x \in \rnx \, | \, \Gamma_i = G(x) \}.
\end{align}

The activation patterns $\Gamma_i$ enable a parametric description of the neural network: 
\begin{align}\label{eq:neural_network_param}
    \begin{split}
        \mathcal{P}(x,\Gamma_i,l;\theta) = \qquad \qquad \qquad \qquad \qquad \qquad \qquad \qquad \qquad\\
        \, \begin{cases}
            f_{l} \circ \gamma_{i,l-1} \odot f_{l-1} \circ \dots \circ \gamma_{i,1} \odot f_1(x), &\text{if } l = L+1, \\
            \gamma_{i,l} \odot f_l \circ \dots \circ \gamma_{i,1} \odot f_1(x), &\text{else}.
        \end{cases}
    \end{split}
\end{align}
which results in an affine function of the state:
\begin{align}\label{eq:parametric_affine_feedback}
\mathcal{P}(x,\Gamma_i,l;\theta) = W_{\Gamma_i,l} x + b_{\Gamma_i,l},
\end{align}
where $W_{\Gamma_i,l} \in \mathbb{R}^{n_l \times n_x}$ and $b_{\Gamma_i,l} \in \mathbb{R}^{n_l}$.
Each activation pattern $\Gamma_i$ defines a possibly non-unique collection of hyperplanes:
\begin{align}\label{eq:computation_hyperplanes}
    D(\Gamma_i) \coloneqq \{-W_{\Gamma_i,l} x \leq -b_{\Gamma_i,l}^{(i)} \, | \, \forall \, \gamma_{i,l}^{(i)} = 1, \forall \, l \in [L]\}.
\end{align}
By stacking all hyperplanes of~\eqref{eq:computation_hyperplanes} and discarding the redundant ones, the description results in the $H$-representation of polytopic regions:
\begin{align}\label{eq:unique_description_region}
    \mathcal{R}_i \coloneqq \{x \in \rnx \, | \,F_i x \leq g_i \},
\end{align}
where $F_i \in \mathbb{R}^{m_i \times n_x}$, $g_i \in \mathbb{R}^{m_i}$ and each row in $F_i$ and $g_i$ corresponds to one of the $m_i \in \mathbb{N}$ non-redundant elements in~\eqref{eq:computation_hyperplanes}, e.g. $F_i^{(j)} = -W_{\Gamma,l}^{(j)}$ and $g_i^{(j)} = -b_{\Gamma,l}^{(j)}$, if $\gamma_l^{(j)} = 1$.

The parametric description of a ReLU network described in this section implies that the network can be seen as an affine function of the input, which includes additional binary variables to describe the ReLU function~\eqref{eq:relu}. This fact will be exploited in the next subsection to formulate a mixed-integer linear program that can compute the set of possible outputs of the network for a given set of inputs.

\section{Verification of neural network-controlled systems via MILP}\label{sec:MILP}
To analyze the safety features of a neural network controller, methods from output range analysis are used.
The goal of output range analysis is to verify if the output of a neural network lies within a desired output set $\mathcal{Y}$ for a given input set $\xin$, i.e. $\mathcal{N}(x_0) \subseteq \mathcal{Y}$ for all $x_0 \in \xin$, where
\begin{align}\label{eq:input_set}
\xin \coloneqq \{ x \in \rnx \, | \, C_{\text{in}} x \leq c_{\text{in}} \},
\end{align}
with $C_{\text{in}} \in \mathbb{R}^{n_{\text{in}} \times n_x}$ and $c_{\text{in}} \in \mathbb{R}^{n_{\text{in}}}$.
Because the output of a neural network controller is per our definition the control input to the system, output range analysis methods can be directly applied to verify the satisfaction of the input constraints.
To do so,
\begin{align}\label{eq:input_constraint_satisfaction}
    C_u \mathcal{N}(x) \leq c_u,  \quad \forall x \in \xin,
\end{align}
must be satisfied for all $n_{cu}$ hyperplanes.
By using the formulation from~\cite{dutta2018output},~\eqref{eq:input_constraint_satisfaction} can be formulated as the following MILPs for each $i \in [n_{cu}]$:
\begin{subequations}\label{eq:MILP_input}
\begin{align}
    & \underset{z,t,u_{0,i},x_0}{\text{maximize}} & & C_u^{(i)} u_{0,i} \label{eq:objective_MILP_input}\\
& \text{subject to} & &  C_{\text{in}} x_0 \leq c_{\text{in}} , \label{eq:initial_set_MILP}\\
&&& z_{0} = x_0, \label{eq:initial_relu_input} \\
&&& u = W_{L+1} z_{L} + b_{L+1}, \label{eq:output_layer} \\
&&& \text{for all } l \in [L]: \nonumber\\
&&& \quad z_{l} \geq W_l z_{l-1} + b_l, \label{eq:relu_geq}\\
&&& \quad z_l \leq W_l z_{l-1} + b_l + M t_l, \label{eq:relu_leq_M}\\
&&& \quad z_{l} \geq 0, \label{eq:relu_geq_0}\\
&&& \quad z_{l} \leq M(\mathbf{1} - t_{l}), \label{eq:relu_0}\\
&&& \quad t_{l} \in [0,1]^{n_l}, \label{eq:cons_binary}
\end{align}
\end{subequations}
where the constraints~\eqref{eq:relu_geq}-\eqref{eq:cons_binary} model the ReLU layers of the neural network controller and~\eqref{eq:output_layer} models the linear output layer.
The constraint~\eqref{eq:initial_set_MILP} guarantees, that $x_0 \in \xin$.
The outputs of each ReLU layer are $z = [z_{0}^T,\dots,z_{L}^T]^T$ and the binary variables are contained in $t = [t_{0}^T, \dots, t_{L}^T]^T$.
The scalar $M \in \mathbb{R}$ has to be larger than the maximum possible output of any neuron $z_{l}^{(i)}$.
By solving~\eqref{eq:MILP_input} globally for all $n_{cu}$ hyperplanes of $\mathcal{U}$, the vector of optimal values of the cost function $c_{u}^* = [C_u^{(0)}u_{0,0}^*, \dots, C_u^{(n_{cu})}u_{0,n_{cu}}^*]^T$ is obtained.
If the resulting set $\mathcal{U}^* \coloneqq \{ y \in \mathbb{R}^{n_y} \, | \, C_{u} y \leq c_{u}^* \}$ provides $\mathcal{U}^* \subseteq \mathcal{U}$, then~\eqref{eq:input_constraint_satisfaction} is guaranteed by construction of the MILPs~\eqref{eq:MILP_input} as proven in~\cite{dutta2018output}.
The process is illustrated on the left side of Fig.~\ref{fig:flowchart_milp}. 

\begin{figure*}[t]
\begin{center}
\includegraphics[width=0.95\textwidth]{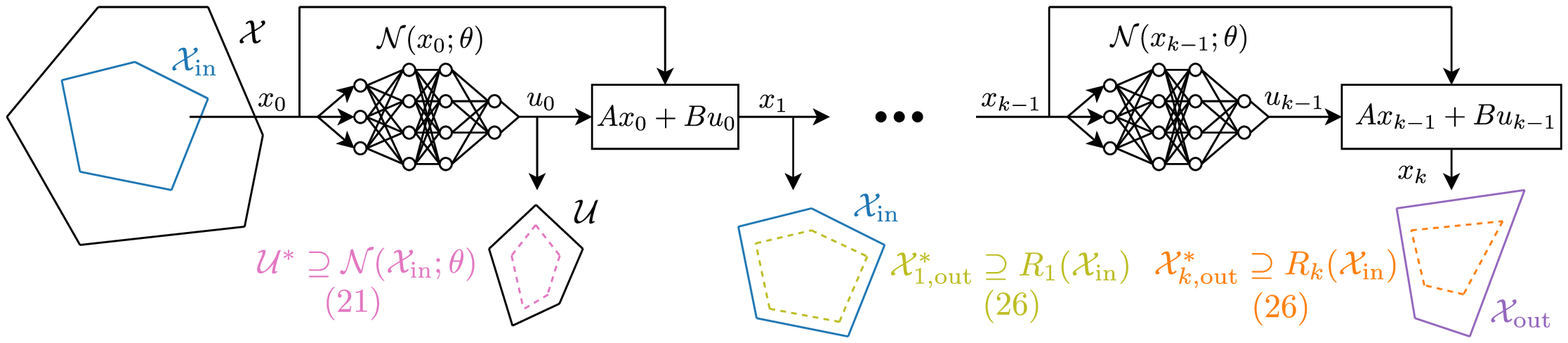}
\caption{Schematic overview of the closed-loop~\eqref{eq:closed_loop} and visualization of the output ranges computed via the MILPs~\eqref{eq:MILP_input} and~\eqref{eq:MILP_state}. The derived sets enable the verification of safety features like closed-loop constraint satisfaction and asymptotic stability.}
\label{fig:flowchart_milp}
\end{center}
\end{figure*}

To analyze further properties of the closed-loop system which are defined in the state space,~\eqref{eq:MILP_input} needs to extended.
We use a set-based formulation of the system evolution to investigate closed-loop constraint satisfaction as well as asymptotic stability.
The behaviour of system~\eqref{eq:lti} controlled by a neural network controller~\eqref{eq:neural_network} is described by:
\begin{align}\label{eq:closed_loop}
    x_{k+1} = f_{\text{cl}}(x_k),
\end{align}
with $f_{\text{cl}}(x_k) = A x_k + B \mathcal{N}(x_k;\theta)$.
We define the $k$-step reachable set as the $k$-time composition of the closed-loop with respect to an initial state set $\xin \subseteq \mathcal{X}$:
\begin{align}\label{eq:k_step_reachable_set}
    R_k(\xin) = f_{\text{cl}} \circ \dots \circ f_{\text{cl}}(\xin).
\end{align}
The desired output set in the state space can be defined by:
\begin{align}\label{eq:input_set}
\mathcal{X}_{\text{out}} \coloneqq \{ x \in \rnx \, | \, C_{\text{out}} x \leq c_{\text{out}} \},
\end{align}
where $C_{\text{out}} \in \mathbb{R}^{n_{\text{out}} \times n_x}$ and $c_{\text{out}} \in \mathbb{R}^{n_{\text{out}}}$.
To verify that the closed-loop system converges to the desired output set, the following condition must be satisfied:
\begin{align}\label{eq:state_constraint_satisfaction_MILP}
    C_{\text{out}} \cdot (f_{\text{cl}} \circ \dots \circ f_{\text{cl}}(x)) \leq c_{\text{out}}, \quad \forall x \in \xin.
\end{align}
By including additional constraints that model the closed-loop of the controlled system to~\eqref{eq:MILP_input}, the following $i \in [n_{\text{out}}]$ MILPs enable the analysis of the  $k$-step reachable set:
\begin{subequations}\label{eq:MILP_state}
\begin{align}
    & \underset{\mathbf{z},\mathbf{t},\mathbf{u},x_0,x_{k,i}}{\text{maximize}} & & C_{\text{out}}^{(i)} x_{k,i} \label{eq:objective_MILP_state}\\
& \text{subject to} & &  C_{\text{in}} x_0 \leq c_{\text{in}} , \\
&&& x_{k,i} = A z_{k,0} + B u_k,  \label{eq:xk_MILP_state} \\
&&& z_{1,0} = x_0, \label{eq:initial_relu_input_state} \\
&&& z_{j,0} = A z_{j-1,0} + B u_{j-1}, \, \forall \, j \in [k] \setminus 1, \label{eq:NN_LTI} \\
&&& \text{for all } j \in [k] \text{ and } l \in [L]: \nonumber\\
&&& \quad z_{j,l} \geq W_l z_{l-1,j} + b_l, \\
&&& \quad z_l \leq W_l z_{l-1,j} + b_l + M t_l, \\
&&& \quad z_{j,l} \geq 0, \\
&&& \quad z_{j,l} \leq M(\mathbf{1} - t_{j,l}), \\
&&& \quad t_{j,l} \in [0,1]^{n_l}, \\
&&& \quad u_j = W_{L+1} z_{j,L} + b_{L+1},
\end{align}
\end{subequations}
where the ReLU outputs and the binary variables are collected in $\mathbf{z} = \{z_0, \dots, z_k \}$ with $z_j = [z_{j,0}^T,\dots,z_{j,L}^T]^T, j \in [k]$, and $\mathbf{t} = \{t_1, \dots, t_k\}$ with $t_j = [t_{j,0}^T, \dots, t_{j,L}^T]^T, j \in [k]$.
The dynamics of the closed-loop system are modelled via the constraints~\eqref{eq:xk_MILP_state}-\eqref{eq:NN_LTI}.
Analogously to~\eqref{eq:MILP_input}, the $n_{\text{out}}$ solutions of~\eqref{eq:MILP_state} define the set $\mathcal{X}_{k,\text{out}}^* \coloneqq\{ x \in \rnx \, | \, C_{\text{out}} x \leq c_{\text{out}}^* \}$ with $c_{\text{out}}^* = [C_{\text{out}}^{(0)}x_{k,0}^*, \dots,C_{\text{out}}^{(n_{\text{out}})}x_{k,n_{\text{out}}}^*]^T$.
If $\mathcal{X}_{k,\text{out}}^* \subseteq \mathcal{X}_{\text{out}}$, then~\eqref{eq:state_constraint_satisfaction_MILP} is satisfied.
Two exemplary usages of~\eqref{eq:MILP_state} are visualized in Fig.~\ref{fig:flowchart_milp}, where the one step reachable set is $R_1(\xin)$ is used in combination with~\eqref{eq:MILP_input} to verify control-invariance in the middle part of the figure. The right part of Fig.~\ref{fig:flowchart_milp} shows how~\eqref{eq:MILP_state} can be used to verify that the k-step reachable set is contained in a desired output set $\mathcal{X}_{\text{out}}$.

\begin{remark}
By choosing the objective functions~\eqref{eq:objective_MILP_input} and~\eqref{eq:objective_MILP_state}, the complexity of the computed output set $\mathcal{Y}^* \coloneqq \{\mathcal{U}^*,\mathcal{X}_{\text{out}}^*\}$, defined by the number of hyperplanes, can be fixed.
Since the solutions of~\eqref{eq:MILP_input} and~\eqref{eq:MILP_state} provide exact bounds for every hyperplane containing the true output set $\mathcal{Y}_{\text{true}} = \{\mathcal{N}(\xin;\theta), R_k(\xin) \}$, if $\mathcal{Y}^*$ and $\mathcal{Y}_{\text{true}}$ have the same complexity and the hyperplanes have the same directions, $\mathcal{Y}_{\text{true}} = \mathcal{Y}^*$ is satisfied.
If the number and directions of the hyperplanes defining $\mathcal{Y}_{\text{true}}$ and $\mathcal{Y}^*$ are not the same, the optimal solution is an over-approximation of the real output set with $\mathcal{Y}_{\text{true}} \subset \mathcal{Y}^*$.
\end{remark}

\section{Safety guarantees}\label{sec:safety}
In the previous section, we proposed an MILP formulation that enables modelling the closed-loop behavior that a neural network-controller, which is usually derived as an approximation of a complex MPC controller.
In this section, we show how the proposed formulation can be used to guarantee constraint satisfaction and convergence to the equilibrium point of the closed-loop system for a set of initial states $\xin$.
\begin{assumption}\label{ass:initial_state_constraints}
    The initial state set $\xin \subseteq \mathcal{X}$ is a subset of the feasible state space.
    This means that for the initial state $x_0 \in \xin$ the state constraints are always satisfied.
\end{assumption}
The derivations of the properties for safety are based on the definition of a control-invariant set for a neural network controlled system.
\begin{definition}\label{def:control_invariant_set}
    A polytope defining a region in the state space:
    \begin{align}\label{eq:control_invariant_set}
        \mathcal{C}\inv \coloneqq \{ x \in \rnx \, | \, C\inv x \leq c\inv\},
    \end{align}
    with $C\inv \in \mathbb{R}^{m\inv \times n_x}$ and $c\inv \in \mathbb{R}^{m\inv}$, is an admissible control-invariant set if: 
    \begin{align}
        \mathcal{N}(\mathcal{C}\inv;\theta) &\subseteq \mathcal{U}, \\
        R_1(\mathcal{C}\inv) &\subseteq \mathcal{C}\inv.
    \end{align}
\end{definition}
If a set is an admissible control-invariant set, all closed-loop trajectories starting at $x_0 \in \mathcal{C}\inv$ will satisfy $x_k \in \mathcal{C}\inv$ for all $k \geq 0$.

\subsection{Guaranteed input constraint satisfaction}
%
By globally solving~\eqref{eq:MILP_input}, the input constraints can be directly verified. 
In addition, if the input constraints can be described by box constraints $\mathcal{U} \coloneqq \{u \in \rnu \, | \, u_{\text{lb}} \leq u \leq u_{\text{ub}} \}$, then the neural network controller can be easily modified to enforce that the input constraints are satisfied for all $x \in \rnx$ as proven in the following Proposition.
\begin{proposition}\label{prop:output_saturation}
    For box constraints $\mathcal{U} \coloneqq \{u \in \rnu \, | \, u_{\text{lb}} \leq u \leq u_{\text{ub}} \}$, each neural network controller $\mathcal{N}(x;\theta)$ with $L$ hidden layers can be adapted such that the modified neural network controller $\mathcal{N}_{\text{sat}}(x) \coloneqq \mathcal{N}(x;\theta_{\text{sat}})$ with $L_{\text{sat}} = L+2$ hidden layers satisfies the box input constraints $\mathcal{U}$ for all $x \in \mathbb{R}^{n_x}$ while all originally feasible input control signals $u_{\text{lb}} \leq \mathcal{N}_{\text{sat}}(x) \leq u_{\text{ub}}$ are left unchanged. 
    The first $L$ elements of the modified controller are equal to the original controller:
    $$\theta_{l,\text{sat}} =\theta_l \quad \forall l \in [L].$$
    The weights of the layers $L$ to $L+3$ are given by:
    \begin{align*}
        W_{L+1,\text{sat}} &= -W_{L+1}, &b_{L+1,\text{sat}} &= u_{\text{ub}} - b_{L+1}, \\
        W_{L+2,\text{sat}} &= -I,       &b_{L+2,\text{sat}} &= u_{\text{ub}} - u_{\text{lb}}, \\
        W_{L+3,\text{sat}} &= I,        &b_{L+3,\text{sat}} &= u_{\text{lb}}.
    \end{align*}
\end{proposition}
\begin{proof}[Proposition~\ref{prop:output_saturation}]
    By substituting the original input signal of the network via $\tilde{u} \coloneqq \mathcal{P}(x,\Gamma_i,L+1;\theta)$, the adapted network results in the function:
    $$s(\tilde{u}) =  \text{max}(-\text{max}(-\tilde{u}+u_{\text{ub}},0)+u_{\text{ub}} - u_{\text{lb}},0)+ u_{\text{lb}}.$$
    The possible outcome of each element $i$ of the function
    $$s(\tilde{u})^{(i)} = 	
            \begin{cases}
        		u_{\text{ub}}^{(i)} & \text{if} \quad \tilde{u}^{(i)} \geq u_{\text{ub}}^{(i)}, \\
        		\tilde{u}^{(i)} & \text{if} \quad u_{\text{lb}}^{(i)} \leq \tilde{u}^{(i)} \leq u_{\text{ub}}^{(i)}, \\
        		u_{\text{lb}}^{(i)} & \text{if} \quad \tilde{u}^{(i)} \leq u_{\text{lb}}^{(i)},
        	\end{cases}
    $$
    represent the box constraints.
    
\end{proof}

\subsection{Guaranteed state constraint satisfaction}%
The analysis of closed-loop state constraint satisfaction relies on the $1$-step reachable set $R_1(\xin) = f_{\text{cl}}(\xin)$.
\begin{lemma}\label{lem:recursive_feasibility}
If Assumption~\ref{ass:initial_state_constraints} holds, the solution of~\eqref{eq:MILP_input} satisfies $\mathcal{U}^* \subseteq \mathcal{U}$ and the solution of~\eqref{eq:MILP_state} results in $\mathcal{X}_{1,\text{out}}^* \subseteq \xin$, $\xin$ is an admissible control-invariant set guaranteeing the satisfaction of state and input constraints for all time steps $k \geq 0$.
\end{lemma}
\begin{proof}[Lemma~\ref{lem:recursive_feasibility}]
    The solution of~\eqref{eq:MILP_input} with $\mathcal{U}^* \subseteq \mathcal{U}$ guarantees that the input constraints are satisfied for all $x \in \xin$.
    By the solution of~\eqref{eq:MILP_state} with $\mathcal{X}_{1,\text{out}}^* \subseteq \xin$, it is guaranteed that $R_1(\xin) \subseteq \xin$, which means that $\xin$ is an admissible control-invariant set by Definition~\ref{def:control_invariant_set}.
    Because $x_0 \in \xin$ with $\xin \subseteq \mathcal{X}$ by Assumption~\ref{ass:initial_state_constraints}, $R_k(\xin) \subseteq \xin$ for all $k \geq 0$, which in return guarantees $\mathcal{N}(R_k(\xin)) \subseteq\mathcal{U}$ for all $k \geq 0$.
    
    %
    
\end{proof}

\subsection{Guaranteed asymptotic stability}

To establish asymptotic stability, it is necessary to prove convergence to the equilibrium $x\eq$.
The set of hyperplanes~\eqref{eq:computation_hyperplanes}
obtained with the equilibrium activation pattern $\Gamma\eq = G(x\eq)$ via \eqref{eq:activation_vector}, results in the unique description~\eqref{eq:unique_description_region} of the polytopic region
\begin{align}\label{eq:equilibrium_region} 
    \mathcal{R}\eq = \{x \in \rnx \, | \, F\eq x \leq g\eq\}.
\end{align}
Following~\eqref{eq:parametric_affine_feedback}, the neural network controller can be represented for all $x \in \mathcal{R}\eq$ as the affine state feedback:
    \begin{align}\label{eq:equilibrium_feedback}
       \mathcal{P}(x,\Gamma\eq,L+1,\theta) &= W_{L+1} (W_{\Gamma\eq,L} x + b_{\Gamma\eq,L}) + b_{L+1}
    \end{align}
\begin{lemma}\label{thm:asymptotic_stability}
If the resulting bias of the neural network is zero, that is:
\begin{align}\label{eq:bias_annihilation}
    W_{L+1} b_{\Gamma\eq,L} + b_{L+1} = 0,
\end{align}
and the resulting weight matrix $W_{L+1} W_{\Gamma\eq}$ satisfies:
\begin{align}\label{eq:prop_stability}
    \lVert\text{eig}(A+B W_{L+1} W_{\Gamma\eq},L)\rVert_{\infty} < 1,
\end{align}
then the neural network controller is asymptotically stabilizing for all $x \in \mathcal{R}_{\text{as}}$, where $\mathcal{R}_{\text{as}}$ is a control invariant set that satisfies the following property:
\begin{subequations}
    \begin{align}\label{eq:definition_stability_set}
        \mathcal{R}_{\text{as}} &\subseteq (\mathcal{R}\eq \cap \mathcal{R}_K),
    \end{align}
\end{subequations}
where $\mathcal{R}_K$ is the region for which the application of the feedback $u = W_{L+1} W_{\Gamma\eq},L) x$ results in asymptotic convergence to the equilibrium $x\eq$ without constraint violations.
\end{lemma}
\begin{proof}[Lemma~\ref{thm:asymptotic_stability}]
    For a linear time-invariant system~\eqref{eq:lti} a state-feedback $u = -Kx$ satisfying
    \begin{align}\label{eq:prop_stability_lti}
    \lVert\text{eig}(A-BK)\rVert_{\infty} < 1
    \end{align}
    leads to asymptotic stable behavior within the corresponding admissible set $\mathcal{R}_K$.
    If a neural network controller admits an equal feedback in the equilibrium region via $K = -W_{L+1} W_{\Gamma\eq,L}$ and $W_{L+1} b_{\Gamma\eq,L} + b_{L+1} = 0$, this would also imply asymptotically stabilizing behaviour of the neural network for all $x \in \mathcal{R}_K$.
    Since the equilibrium feedback is only applied within the equilibrium region, asymptotic stability can only be guaranteed for a control-invariant set $\mathcal{R}_{\text{as}}$ within the intersection of the equilibrium region $\mathcal{R}\eq$ and the  region $\mathcal{R}_K$.
    
\end{proof}

\begin{remark}\label{rem:zero_bias}
 If the neural network controller $\mathcal{N}(x;\theta)$ has zero-bias, i.e. $b_l = 0$ for all $l=[L+1]$, then~\eqref{eq:bias_annihilation} is always satisfied.
\end{remark}

\begin{theorem}\label{cor:asymptotic_stability_for_xin}
    If the neural network controller satisfies~\eqref{eq:bias_annihilation} and~\eqref{eq:prop_stability} and the solution of~\eqref{eq:MILP_state} provides for a chosen time-step $k$ that $\mathcal{X}_{k,\text{out}}^* \subseteq \mathcal{R}_{\text{as}}$, then the closed-loop system~\eqref{eq:closed_loop} is asymptotically stable for all $x \in \xin$.
\end{theorem}
\begin{proof}[Theorem~\ref{cor:asymptotic_stability_for_xin}]
    Because~\eqref{eq:bias_annihilation} and~\eqref{eq:prop_stability} are satisfied, asymptotic stability for all $x \in \mathcal{R}_{\text{as}}$ follows from Lemma~\ref{thm:asymptotic_stability}.
    Since $R_k(\xin) \subseteq \mathcal{X}_{k,\text{out}}^*$ and $\mathcal{X}_{k,\text{out}}^* \subseteq \mathcal{R}_{\text{as}}$, $x_k \in \mathcal{R}_{\text{as}}$ at least after $k$ closed-loop steps~\eqref{eq:closed_loop} for all $x \in \xin$, from which follows asymptotic stability for all $x \in \xin$.
    
\end{proof}

If the feedback defined by the neural network in the region $\mathcal{R}\eq$ is equal to the LQR controller, then the neural network controller behaves optimally in the neighborhood of the equilibrium point.
This is formalized in the following result.
\begin{corollary}\label{cor:lqr_asymptotic_stability}
    If the feedback of the neural network controller in the neighborhood of the equilibrium~\eqref{eq:equilibrium_feedback} is equal to the LQR feedback:
    $$\mathcal{P}(x,\Gamma\eq,L+1;\theta) = -K\lqr x $$
    and the solution of~\eqref{eq:MILP_state} provides $\mathcal{X}_{k,\text{out}}^* \subseteq \mathcal{R}_{\text{as}}$,
    then the neural network controller drives the system optimally w.r.t~\eqref{eq:objective_oc} to the equilibrium for all $x \in \mathcal{R}_{\text{as}}$ and asymptotically to the equilibrium for all $x \in \xin$.
\end{corollary}
\begin{proof}[Corollary~\ref{cor:lqr_asymptotic_stability}]
    The LQR state feedback~\eqref{eq:lqr_feedback} is the optimal solution w.r.t~\eqref{eq:objective_oc} for all $x \in \mathcal{R}\lqr$.
    If $\mathcal{P}(x,\Gamma\eq,L+1;\theta) = -K\lqr x$, the neural network controller provides the LQR feedback for all $x \in \mathcal{R}\eq$.
    Hence, the neural network controller returns an LQR optimal control input w.r.t~\eqref{eq:objective_oc} for all $x \in (\mathcal{R}\eq \cap \mathcal{R}\lqr)$.
    Because the LQR feedback is a special case of an asymptotically stabilizing feedback:
    \begin{align}
        K\lqr \in \mathcal{K} \coloneqq \{K \in \mathbb{R}^{n_u \times n_x} \, | \,\eqref{eq:prop_stability_lti} \},
    \end{align}
    the proof for asymptotic stability and convergence to the stability set $\mathcal{R}_{\text{as}}$ is analogous to the proof of Lemma~\ref{thm:asymptotic_stability} and Theorem~\ref{cor:asymptotic_stability_for_xin} by substituting $K$ with $K\lqr$ and $\mathcal{R}_K$ with $\mathcal{R}\lqr$.
    
\end{proof}
In general, a neural network controller does not satisfy the requirements~\eqref{eq:bias_annihilation} and~\eqref{eq:prop_stability} for asymptotic stability.
The next section shows an optimization-based method to ensure that such requirements are satisfied.

\section{LQR-optimized neural network controller}\label{sec:NN_LQR}

In this section, an optimization-based method is presented to modify a neural network such that it provides the same feedback as an LQR controller in the equilibrium region without changing the regions $\mathcal{R}_i$ and without retraining the neural network.
The goal is that the modified controller satisfies the conditions of Theorem~\ref{cor:lqr_asymptotic_stability}.
The main idea is to adapt the values of the weight and the bias in linear output layer $L+1$, because the regions $\mathcal{R}_i$ are only depending on the $L$ hidden layers, as~\eqref{eq:computation_hyperplanes} shows.
\begin{lemma}\label{lem:optimal_lqr_adaptation}
For every neural network controller $\mathcal{N}(x;\theta)$, it is possible to find values for the weights of the final layer, such that
\begin{align}\label{eq:optimal_solution_existence}
    \hat{W}_{L+1} W_{\Gamma\eq,L} = -K\lqr,
\end{align}
if the following requirements are satisfied:
\begin{subequations}\label{eq:existence_solution_lqr_adaptation}
    \begin{align}
        \text{rank}(A\eq) &= \text{rank}(\begin{bmatrix}A\eq &b\eq\\ \end{bmatrix}), \label{eq:existence_solution_lqr_adaptation_ranks} \\
        n_u n_{L} &\geq \text{rank}\left(A\eq\right), \label{eq:existence_lqr_adaptation_variables}
    \end{align}
\end{subequations}
where
\begin{align*}
    A\eq = \begin{bmatrix}
                    W_{\Gamma\eq,L}^T &0 &0 \\
                    0 &\ddots &0 \\
                    0 &0 &W_{\Gamma\eq,L}^T \\
                \end{bmatrix},
    b\eq =\begin{bmatrix}
                    -{K\lqr^{(i)}}^T \\
                    \vdots \\
                    -{K\lqr^{(n_u)}}^T \\
                \end{bmatrix}.
\end{align*}
\end{lemma}
\begin{proof}[Lemma~\ref{lem:optimal_lqr_adaptation}]
    Equation~\eqref{eq:optimal_solution_existence} can be reformulated as a system of linear equations $A\eq w = b\eq$ with $w = [\hat{W}_{L+1}^{(1)}, \dots, \hat{W}_{L+1}^{(n_u)}]^T \in \mathbb{R}^{n_u n_{L}}$.
    A system of linear equations admits at least one solution if conditions~\eqref{eq:existence_solution_lqr_adaptation} are satisfied.
    
\end{proof}
We can now state another contribution of our paper in which we propose a systematic adaptation of the weights of the last layer to ensure that LQR performance is achieved in a region around the equilibrium point.
\begin{theorem}\label{cor:optimal_lqr_shift}
If conditions~\eqref{eq:existence_solution_lqr_adaptation} are satisfied, the solution of the following convex optimization problem: 
\begin{subequations}\label{eq:lqr_adaptation}
\begin{align}
    & \underset{\hat{W}_{L+1},\hat{b}_{L+1}}{\text{minimize}} & & \sum_{i,j = 1}^{n_u,n_{L}}  (\hat{W}_{L+1}^{(i,j)} - W_{L+1}^{(i,j)})^2  + \sum_{i=1}^{n_u} (\hat{b}_{L+1}^{(i)} - b_{L+1}^{(i)})^2 \label{eq:lqr_adaptation_objective}\\
& \text{subject to} & & \hat{W}_{L+1} W_{\Gamma\eq} = -K_{\text{lqr}}, \label{eq:lqr_K}\\
&&& \hat{W}_{L+1} b_{\Gamma\eq} + \hat{b}_{L+1} = 0, \label{eq:lqr_zero_bias}
\end{align}
\end{subequations}
provides the weight $\hat{W}_{L+1}^*$ and bias $\hat{b}_{L+1}^*$ for the last layer, such that $\mathcal{N}(x;\theta\lqr) = -K\lqr x$ for all $x \in \mathcal{R}\eq$ with $\theta_{l,\text{lqr}} = \theta_l$ for $l = [L]$ and $\theta_{L+1,\text{sat}} = \{\hat{W}_{L+1}^*,\hat{b}_{L+1}^*\}$ while minimizing the change in parameters of the last layer. 
\end{theorem}
\begin{proof}[Theorem~\ref{cor:optimal_lqr_shift}]
    If the conditions~\eqref{eq:existence_solution_lqr_adaptation} are satisfied, Lemma~\ref{lem:optimal_lqr_adaptation} guarantees that~\eqref{eq:lqr_K} can always be satisfied.
    By setting $\hat{b}_{L+1} = -\hat{W}_{L+1}b_{\Gamma\eq}$,~\eqref{eq:lqr_zero_bias} is satisfied for every $\hat{W}_{L+1}$.
    This means that the feedback of the neural network controller in the equilibrium region is equal to the LQR feedback if the weight and the bias of the last layer are given by $\hat{W}_{L+1}^*$ and bias $\hat{b}_{L+1}^*$.
    Since the objective function~\eqref{eq:lqr_adaptation_objective} describes the change in the parameters of the last layer and the optimization is convex, the optimal solution guarantees that the change is minimal.
    
\end{proof}


\section{Case study: Double integrator}\label{sec:case_study}

\subsection{Control problem}
For the demonstration and visualization of the proposed approach, the discrete-time double integrator is considered with two states $x = [s, v]^T$, where $s$ is the position and $v$ is the velocity, and one input $u = a$, where $a$ is the acceleration.
The system matrices~\eqref{eq:lti} are derived via Euler discretization with a sampling time of $\Delta t =  \SI{0.1}{\second}$:
\begin{align*}
&A=
\begin{bmatrix}
0.5403   & -0.8415 \\ 
0.8415  & \hm0.5403 \\
\end{bmatrix},
&B=
\begin{bmatrix}
-0.4597 \\
\hm0.8415 \\
\end{bmatrix}.
\end{align*}
The polytopic state constraints $\mathcal{X}$ are given by
\begin{small}
\begin{align*}
&C_x=
\begin{bmatrix}
\hm1   & \hm0 \\ 
-1  & \hm0 \\
\hm0   & \hm1 \\
\hm0   & -1 \\
\end{bmatrix},
&c_x=
\begin{bmatrix}
5.0 \\
5.0 \\
5.0 \\
5.0 \\
\end{bmatrix},
\end{align*}
\end{small}
and the box constraints for the input are $u_{\text{lb}} = -1$ and $u_{\text{ub}} = 1$.
The weight matrices for the objective~\eqref{eq:objective_oc} are $Q = 2I$ and $R=1$.
The initial state space $\xin$ was computed with~\cite{herceg2013multi} for~\ref{eq:mpc} with $N=3$:

\begin{small}
\begin{align*}
&C_{\text{in}} =
\begin{bmatrix}
\hm0.0707  & -0.9975 \\ 
-0.1509  & -0.9885 \\ 
-0.8011  & -0.5984 \\
-0.9797  &  \hm0.2004 \\
\hm0.8776  & -0.4795 \\
\hm0.9797  & -0.2004 \\
\hm0.8012  &  \hm0.5984 \\
\hm0.1509  &  \hm0.9885 \\
-0.0707  &  \hm0.9975 \\
-0.8776  &  \hm0.4754 \\
\end{bmatrix},
&c_{\text{in}} =
\begin{bmatrix}
3.0297 \\
2.9401 \\
3.5051 \\
3.2918 \\
3.3082 \\
3.2918 \\
3.5051 \\
2.9401 \\
3.0297 \\
3.3082 \\
\end{bmatrix}.
\end{align*}
\end{small}

\subsection{LQR optimized neural network controller}

We trained a neural network with $L=1$ hidden layers and $n_L = 10$ neurons via imitation learning~\cite{karg2018efficient} with:
The network and the regions it defines via~\eqref{eq:unique_description_region} within $\xin$ are illustrated in Fig.~\ref{fig:neural_network}.
\begin{figure}[t]
\begin{center}
\includegraphics[width=0.99\columnwidth]{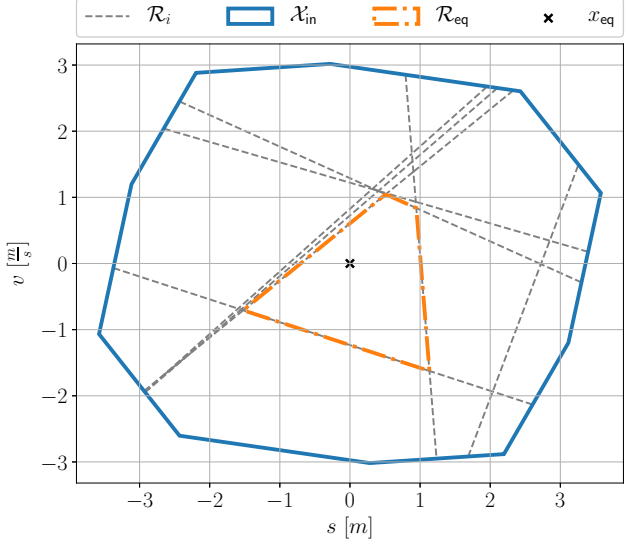}
\caption{Regions $\mathcal{R}_i$ defined by the neural network controller within $\xin$. The equilibrium region $\mathcal{R}\eq$ containing the equilibrium $x\eq$ is highlighted.}
\label{fig:neural_network}
\end{center}
\end{figure}
To find the weights and bias of the final layer providing LQR feedback for all $x \in \mathcal{R}\eq$, problem~\eqref{eq:lqr_adaptation} was solved with $K\lqr = \begin{bmatrix} 0.2501 &0.8290\\ \end{bmatrix}$.
The LQR adaptation~\eqref{eq:lqr_adaptation} was feasible with an optimal cost of \num{5.4836e-4}.
In Fig.~\ref{fig:trajectory} the performance of the two controllers $\mathcal{N}_{\text{sat}}(x)$ and $\mathcal{N}_{\text{sat,lqr}}(x)$, which were modified via Proposition~\ref{prop:output_saturation} to satisfy the box input constraints, is compared.
The closed-loop trajectories are very similar, as the low cost for LQR adaptation suggested, but only $\mathcal{N}_{\text{sat,lqr}}(x)$ drives the system asymptotically to the equilibrium, which can be seen in the magnified area. Note that our proposal solves an important problem of approximate MPC controllers as there is no set-point tracking error.
\begin{figure}[t]
\begin{center}
\includegraphics[width=0.99\columnwidth]{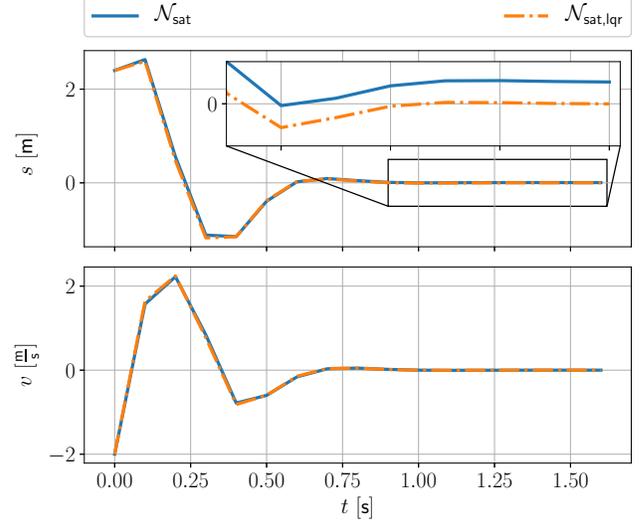}
\caption{Exemplary state trajectories of the closed-loop systems~\eqref{eq:closed_loop} for the LQR optimized controller $\mathcal{N}_{\text{sat},\text{lqr}}(x)$ and the original saturated controller $\mathcal{N}_ {\text{sat}}(x)$ for the same initial state $x_0 \in \xin$.}
\label{fig:trajectory}
\end{center}
\end{figure}

\subsection{Closed-loop constraint satisfaction}

In order to guarantee constraint satisfaction, the state and input constraints of the system need to be satisfied for all time steps $k$.
The one-step reachability analysis via~\eqref{eq:MILP_state} for the initial state space $\xin \subseteq \mathcal{X}$ guarantees that $\mathcal{X}_{1,\text{out}}^* \subset \xin$, as Fig.~\ref{fig:state_space} shows.
This means that $\xin$ is a control-invariant set and closed-loop constraint satisfaction is guaranteed for the saturated controller $\mathcal{N}_{\text{lqr,sat}}(x)$ by Lemma~\ref{lem:recursive_feasibility}.
\subsection{Asymptotic stability}
Asymptotic stability can be guaranteed if the closed-loop system converges to $\mathcal{R}_{\text{as}} \subseteq (\mathcal{R}\lqr \cap \mathcal{R}\eq)$ for all $x \in \xin$.
The stability set $\mathcal{R}_{\text{as}}$ is computed with~\cite{herceg2013multi} as the maximum control-invariant set with state constraints $x \in (\mathcal{R}\lqr \cap \mathcal{R}\eq)$ and given by $\mathcal{R}_{\text{as}} \coloneqq \{x \in \rnx \, | \, C_{\text{as}} x \leq c_{\text{as}} \}$ with
\begin{small}
\begin{align*}
&C_{\text{as}}=
\begin{bmatrix}
-0.3264  &   -0.9452 \\
-0.6533  &    \hm0.7571 \\
-0.2889  &   -0.9574 \\
\hm0.9971  &    \hm0.0759 \\
\hm0.8301  &   -0.5576 \\
-0.2888  &   -0.9574 \\
\hm0.4307  &    \hm0.9025 \\
\end{bmatrix},
&c_{\text{as}}=
\begin{bmatrix}
 1.1657 \\
 0.4590 \\
 1.1548 \\
 1.0070 \\
 1.1920 \\
 1.1549 \\
 1.1637 \\
\end{bmatrix}.
\end{align*}
\end{small}
The equilibrium region~\eqref{eq:equilibrium_region} of the controller is given by:
\begin{small}
\begin{align*}
&F\eq=
\begin{bmatrix}
-0.2527 & -0.7318 \\
\hm0.2646 & \hm0.0201 \\
-0.3536 &  \hm0.4097 \\
\hm0.3115 &  \hm0.6526 \\
\end{bmatrix},
&g\eq=
\begin{bmatrix}
0.9025 \\
0.2673 \\
0.2484  \\
0.8415
\end{bmatrix},
\end{align*}%
\end{small}%
and the LQR admissible set, computed with~\cite{herceg2013multi}, is defined by the polytope:
\begin{small}
\begin{align*}
&F\lqr=
\begin{bmatrix}
-0.6870  &  \hm0.24566 \\
\hm0.6870  & -0.2456 \\
-0.2501  & -0.8290 \\
\hm0.2501 &  \hm0.8290 \\
\end{bmatrix},
&g\lqr=
\begin{bmatrix}
1 \\
1 \\
1 \\
1 \\
\end{bmatrix}.
\end{align*}
\end{small}
The presented sets are all visualized in Fig.~\ref{fig:state_space}. 
By solving~\eqref{eq:MILP_state} with $k = 6$, we obtain $\mathcal{X}_{6,\text{out}}^* \subseteq \mathcal{R}_{\text{as}}$.
Because $R_6(\xin) \subseteq \mathcal{X}_{6,\text{out}}^*$ we can conclude asymptotic stability for all $x \in \xin$ according to Corollary~\ref{cor:lqr_asymptotic_stability}.
\begin{figure}[!t]
\begin{center}
\includegraphics[width=0.95\columnwidth]{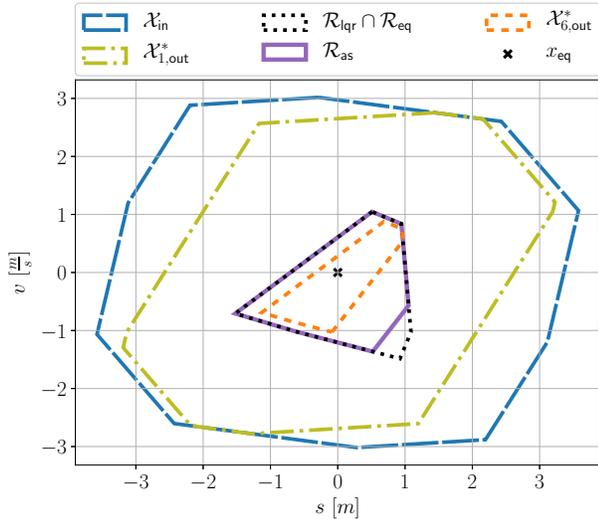}
\caption{Overview of the sets in the state space resulting from the process of verifying closed-loop constraint satisfaction and asymptotic stability via~\eqref{eq:MILP_state}.}
\label{fig:state_space}
\end{center}
\end{figure}
\section{Conclusion and future work}\label{sec:conclusions}
We have presented requirements for neural network controllers such that they guarantee asymptotically stable behaviour and lead to closed-loop constraint satisfaction.
These requirements can be verified via a mixed-integer linear programming scheme that represents the neural network controller.
Additionally, we have proposed an optimization-based modification of neural network controllers such that they satisfy the requirements necessary for asymptotically stabilizing behaviour. This modification can be performed after training the neural networks, avoiding the computationally expensive retraining of the controllers.

Future work includes the use of SDP relaxations instead of an MILP formulation to compute the sets that guarantee the safe application based on the neural network controller.

\bibliography{IEEEabrv,cdc_2020}

\end{document}